\documentclass{article} \usepackage{amsthm} \usepackage{amsmath}

\usepackage{mathtools}

\usepackage{amssymb}
\usepackage{bbm}
\usepackage{xcolor}
\usepackage{comment}
\usepackage{tikz, tikz-3dplot}
\usepackage[T1]{fontenc}
\usepackage[linesnumbered,lined,commentsnumbered,noend,boxed]{algorithm2e}

\SetCommentSty{FontInPseudocodeComments}

\usetikzlibrary{arrows}
\pgfarrowsdeclarecombine*{spaced stealth'}{spaced stealth'}{stealth'}{stealth'}{space}{space}
\tikzset{>=spaced stealth'}

\newcommand*{\T}{^{\mkern-1.5mu\mathsf{T}}}

\usepackage{a4}
\usepackage{todonotes} 
\usepackage{utf8math}
\usepackage{enumitem} 
\usepackage{mathrsfs}

\newcommand{\norm}[1]{  \lVert #1 \rVert}
\newcommand{\ceil}[1]{  \lceil #1 \rceil}
\newcommand{\floor}[1]{\lfloor #1\rfloor} 

\DeclareMathOperator{\poly}{poly}
\DeclareMathOperator{\polylog}{polylog}
\DeclareMathOperator{\supp}{supp}

\newcommand{\Z}{\mathbb{Z}}
\newcommand{\ZZ}{\Z_{\geq 0}}
\newcommand{\R}{\mathbb{R}}
\newcommand{\RR}{\R_{\geq 0}}
\newcommand{\lex}{\prec_{\mathrm{lex}}}

\usepackage[capitalize]{cleveref}

\newtheorem{theorem}{Theorem}
\newtheorem{lemma}[theorem]{Lemma}

\newtheorem{definition}[theorem]{Definition}

\title{Space-Efficient Algorithm for Integer Programming with Few Constraints}

\author{Lars Rohwedder\footnote{Maastricht University, Maastricht, Netherlands. Supported by Dutch Research Council (NWO) project “The Twilight
Zone of Efficiency: Optimality of Quasi-Polynomial Time Algorithms” [grant number OCEN.W.21.268]} \and Karol
W\k{e}grzycki\footnote{Saarland University and Max Planck Institute for Informatics,
        Saarbr\"ucken, Germany. 
    This work is part of the project TIPEA that has
    received funding from the European Research Council (ERC) under the European Unions Horizon
2020 research and innovation programme (grant agreement No. 850979).}
}

\date{}

\begin{document}
\maketitle
\begin{abstract}
  \noindent
  Integer linear programs $\min\{c\T x : A x =
  b, x \in \mathbb{Z}^n_{\ge 0}\}$, where $A \in \mathbb{Z}^{m \times n}$, $b
  \in \mathbb{Z}^m$, and $c \in \mathbb{Z}^n$, can be
  solved in pseudopolynomial time for any fixed number of constraints $m = O(1)$.
  More precisely, in time $(m\Delta)^{O(m)} \poly(I)$,
  where $\Delta$ is the maximum absolute
  value of an entry in $A$ and $I$ the input size.

  Known algorithms rely heavily on dynamic programming, which
  leads to a space complexity of similar order of magnitude as the running time.
  In this paper, we present a polynomial space algorithm that solves integer linear programs in
	$(m\Delta)^{O(m (\log m + \log\log\Delta))} \poly(I)$ time, that is,
  in almost the same time as previous dynamic programming algorithms.
\end{abstract}

\section{Introduction}

We consider integer linear programs (ILP) in equality form
\begin{equation}\label{eq:ilp}
		\min \left\{ c\T x : A x = b, x \in \ZZ^n \right\}
\end{equation}
where $A \in \Z^{m \times n}$, $b \in \Z^m$, and $c \in \Z^n$. 
Integer linear programs are a widely used model to describe problems in combinatorial optimization and
the algorithms available for them are used as general purpose optimization suites.
Unlike its continuous counterpart, where $x\in \RR^n$, integer linear programming is NP-complete even in the case of $m=1$. Hence, theoretical research is mainly focused on identifying structural assumptions, under which ILPs can still be
solved efficiently. Here, two dominant directions can be distinguished.

Based on the geometry of numbers, Lenstra~\cite{lenstra1983integer} showed that ILPs of the form $\min\{c\T x : A x \le b, x\in \Z^n\}$ can be solved in polynomial time for constant $n$, in fact, that they are fixed-parameter tractable (FPT) in parameter $n$. 
Note that this form is equivalent to~\eqref{eq:ilp} by a standard transformation, but this would led to an increase in
parameter $n$.
Subsequent improvements by Kannan~\cite{kannan1983improved}
and most recently Reis and Rothvoss~\cite{reis23} lead to the currently best running time of $(\log n)^{O(n)} \poly(I)$, where $I$ is the input size.

The second direction is mainly based on dynamic programming approaches and goes
back to Papadimitriou~\cite{papadimitriou1981complexity}. He considered~\eqref{eq:ilp} in the case when the absolute values of the entries of $A$ and $b$ are bounded by $\Delta$ and designed a pseudopolynomial time algorithm when $m$ is fixed. This generalizes the classical pseudopolynomial time algorithm for Knapsack problems~\cite{bellman1966dynamic}. Later, his algorithm was improved by Eisenbrand and Weismantel~\cite{eisenbrand2019proximity}. The current best algorithm is due to Jansen and Rohwedder~\cite{jansen23} with a running time of $O(\sqrt{m}\Delta)^{2m} \poly(I)$.
In the latter two results the bound of $\Delta$ is only required on $A$, but not $b$.
Similar approaches were also applied successfully to constraint matrices $A$, where both dimensions $n$ and $m$ can
be large but satisfy certain structures of their non-zero entries, see e.g.~\cite{CslovjecsekEHRW21}.

\subparagraph*{Space efficiency of integer linear programs.}
NP-completeness tells us that any algorithm that solves general ILPs needs to have an exponential running
time at least in some parameter unless P$=$NP. But not only running time, also space usage can be a limiting factor.
It is well known that for practical ILP solvers based on Branch-and-Bound it is often a bottleneck to store
the list of active nodes, see e.g.~\cite[Section~7.5.1]{wolsey2020integer}.
This motivates the theoretical question of whether the aforementioned
results can also be achieved in polynomial space or, more generally, what tradeoffs between running
time and space one can achieve.

Frank and Tardos~\cite{frank1987application} showed that Lenstra's algorithm and the improvement by Kannan
with running time $n^{O(n)}\poly(I)$ can be implemented in polynomial space.
For the improvement by Reis and Rothvoss this is not known.

Regarding the algorithms based on dynamic programming, the space usage is also pseudopolynomial, which
is inherent in the method. Alternative, space-efficient methods are known for Subset Sum, a very
restricted variant of ILP, where the task is to find a solution to
$\{x \in \{0,1\}^n : a\T x = b \}$ with $a \in \ZZ^n$ and $b \in \ZZ$. Note that here $m=1$ and there is
no objective function.
For the Subset Sum problem, Lokshtanov and Nederlof~\cite{algebraization10} introduced an algebraization technique and designed an $\tilde{O}(n^3 b \log{b})$ time and $\tilde{O}(n^2)$ space algorithm (see also~\cite{kane2010unary,elberfeld2010logspace}). This algorithm was later improved by Bringmann~\cite{DBLP:conf/soda/Bringmann17}, and the current best low-space algorithm for Subset Sum is due to Jin et al.~\cite{jin2021fast}, which runs in $\tilde{O}(nb)$ time and $O(\log{n} \log\log{n} + \log{b})$ space.
Despite the intensive research on polynomial space algorithms for Subset Sum, little is known for
more general problems.
The seminal result of Lokshtanov and Nederlof~\cite{algebraization10} is the only one that considers the more general Knapsack problem, that is, $\max\{c\T x : a\T x \le b, x\in\{0,1\}^n\}$ with $c, a\in\ZZ^n$, and $b\in \ZZ$. Nevertheless, their method suffers from pseudopolynomial factor in both $b$ and $\norm{c}_{\infty}$ in the running time, see~\cite[Theorem 4.5]{algebraization10}, whereas the classical dynamic programming algorithm only requires pseudopolynomial time in one of the two parameters. A $\poly(n, b)$ time and polynomial space algorithm is widely open for the Knapsack problem and Unbounded Knapsack (the variant with $x\in \ZZ^n$)~\cite{nederlof2024}.

Savitch's theorem can be used to avoid the high space usage of dynamic programming, but often comes at a significant
cost for the running time.
In the case of~\eqref{eq:ilp}, this method leads to an algorithm in time $(m\Delta)^{O(m^2 \log(m\Delta))} \poly(I)$
and polynomial space, see \Cref{sec:techniques} for details.

\subparagraph*{Our contribution.}
In this paper, we show that integer programming~\eqref{eq:ilp} can be solved in polynomial space
and a running time that nearly matches the best known running time of the dynamic programming approach by~\cite{eisenbrand2019proximity, jansen23}.
\begin{theorem}\label{thm:main}
	Let $A \in \Z^{m \times n}$, $b \in \Z^m$, and $c \in \Z^n$ and
	suppose that the absolute value of each entry of $A$ is bounded by
	$\Delta > 1$. Then, in time 
	\begin{equation*}
		(m\Delta)^{O(m(\log{m}+ \log\log{\Delta}))} \cdot \poly(I)
	\end{equation*}
	and polynomial space, one can compute a solution
	to~\eqref{eq:ilp} or conclude that such a solution does not exist.
	Here, $I$ is the encoding size of $A, b$, and $c$.
\end{theorem}
Compared to the running time in~\cite{eisenbrand2019proximity, jansen23}, this is an additional factor of
$O(\log m + \log\log\Delta)$ in the exponent of $m\Delta$, whereas all known techniques lead to an additional factor of at
least $O(m\log\Delta)$, see \Cref{sec:techniques}.

\section{Overview of previous and new techniques} \label{sec:techniques}
Let us briefly introduce some notation and the model of computation.
Throughout this paper, we assume all logarithms are base 2. We write $\tilde{O}(f) = O(f \cdot \polylog(f))$. Further, we make use of the shorthand notation $[n] = \{1, \ldots, n\}$. For a vector $x \in \Z^n$, let $\norm{x}_p$ denote the $\ell_p$ norm of vector $x$. In the special case of $p = 0$, recall that $\norm{x}_0 = |\supp(x)|$, where $\supp(x) = \{ i \in [n] : x_i \neq 0\}$ are the non-zero coefficients.

For two vectors $x,y \in \ZZ^n$, we say that $x$ is lexicographically smaller than $y$ ($x \lex y$) if there exists an index $i \in [n]$ such that
(i) $x_j = y_j$ for every $j < i$, and (ii) $x_i < y_j$.

As in related works, see~\cite{jin2021fast}, we work with a random-access model with word length $\Theta(\log{n} + \log{W})$, where $W$ is the largest absolute value of an input integer. We measure space complexity in terms of the total number of words in the working memory.
We say that an algorithm runs in polynomial space if the space complexity is bounded by a polynomial in the
input size, that is, in $\poly(I) := \poly(n,m,\log(\Delta), \log(\norm{b}_\infty), \log(\norm{c}_\infty))$.
Similarly, the running time of an algorithm is measured in the number of arithmetic operations
on words.

Before we explain the ideas behind~\cref{thm:main} we will briefly review
previous techniques used to achieve polynomial space dependency.

The standard tool to reduce space complexity usage are the ideas in Savitch's theorem, which transforms
algorithms with non-deterministic space complexity to slightly higher deterministic space complexity.
Here, the key technique is to decide $(s, t)$-reachability in a directed graph with $N$ vertices in
$N^{O(\log{N})}$ time and $O(\log N)$ space. This is straight-forward by recursively
guessing the middle vertex of the path from $s$ to $t$. The space complexity comes from the $O(\log N)$
vertices stored in the stack.
This can easily be adapted to finding the shortest
path in a weighted graph in the same time and space.
One can already apply this to ILP. Namely,
the algorithm of Eisenbrand and Weismantel~\cite{eisenbrand2019proximity}
constructs a directed acyclic graph with $N = (m \Delta)^{O(m)} \cdot \norm{b}_{\infty}$ vertices and
then computes the shortest path between two of the vertices. 
With standard polynomial time preprocessing described in~\cite{jansen23}, one can reduce $\norm{b}_{\infty}$ to $(m\Delta)^{O(m)}$.
This shortest path query is the main source of the high space usage of~\cite{eisenbrand2019proximity} and
using Savitch's theorem, the ILP can therefore be solved in $N^{O(\log{N})} =
(m\Delta)^{O(m^2 \log(m\Delta))}$ time and $O(\log{N}) = O(m \log(m\Delta))$ space.

A more intricate technique is the algebraization technique~\cite{algebraization10}, which relies on the following coefficient test lemma.
\begin{theorem}[cf.\ Lemma 3.2 in~\cite{jin2021fast}]\label{lem:test}
	Consider a polynomial $f(x) = \sum_{i=0}^d \alpha_i x^i \in \mathbb F[x]$ of degree at most $d$
	and assume there is an algorithm running in time $T$ and space $S$ for
	evaluating $f(a)$ for a given $a \in \mathbb{F}$. Then, there
	is an algorithm running in $\tilde O(d T)$ time and $\tilde O(S)$ space
	that for a given $t \in \{0,\ldots,d\}$ decides if $\alpha_t = 0$. 
\end{theorem}
The idea to solve the Subset Sum problem, is to consider the polynomial $p(x) =
\prod_{i=1}^n (1+x^{a_i})$ and use~\cref{lem:test} to test whether the coefficient
$\alpha_t$ is nonzero, where $a_1,\ldots,a_n \in \ZZ$ are the input integers,
$t \in \ZZ$ is the target, and $\alpha_t$ is a coefficient in front of monomial
$x^t$ of polynomial $p$. Note that the polynomial $p(x)$ can be evaluated in
$\poly(n, \norm{a}_{\infty})$ time, which, combined with~\cref{lem:test} yields the result
of~\cite{algebraization10}. 

This technique can be applied to check feasibility of~\eqref{eq:ilp}
by aggregating the $m$ constraints of $A$ into a single one $a\T x = b'$ where
$\norm{a}_{\infty} = (m\Delta)^{O(m^2)}$, see e.g.~\cite{kannan1983polynomial},
and reducing the non-negative integer variables to binary ones,
which can be done by applying proximity techniques~\cite{eisenbrand2019proximity}.
We omit the details here, since
the method has two inherent disadvantages for our application:
it has a higher running time with an additional factor of $O(m)$ in the exponent compared to~\cite{eisenbrand2019proximity,jansen23} and
to use~\cref{lem:test} to solve
the optimization variant in~\eqref{eq:ilp} one would need to additionally encode the
coefficient of the objective function $c$ in $a$, which would yield a pseudopolynomial running time
in $\norm{c}_\infty$ that is not necessary in~\cite{eisenbrand2019proximity,jansen23}.

Another straight-forward approach is based on the notable property of~\eqref{eq:ilp} that it always has sparse optimal solutions.
\begin{theorem}[cf.\ Corollary 5 in~\cite{caratheodory06}]\label{caratheodory}
	If~\eqref{eq:ilp} is bounded and feasible, then there exists an optimal solution $x^\ast \in \ZZ^n$ such that
	\begin{displaymath}
		\norm{x^\ast}_0 \le  \Gamma,
	\end{displaymath}
	where $\Gamma := 2(m+1)(\log(m+1) + \log{\Delta} + 2)$. Moreover $x^\ast$ is the lexicographically minimal solution.\footnote{Eisenbrand and Shmonin~\cite{caratheodory06} did not state that $x^\ast$ is the lexicographically minimal solution, but this can be derived easily: 
	one can replace the objective by $c'_i = c_i + \epsilon^i$ for each $i\in [n]$, where $\epsilon > 0$ is sufficiently small. This simulates a lexicographical tie-breaking rule.}
\end{theorem}
This property is sometimes used in the following way: one can guess the support of an optimal solution
from the $O(n^\Gamma)$ potential variable sets and then apply Lenstra's algorithm for integer linear programming
in small dimension. This idea was used for example in~\cite{jansen2010eptas}.
Recall that Lenstra's algorithm and the improvement by Kannan can be implemented in polynomial space.
\begin{theorem}[cf.\ Theorem~5.2 in~\cite{frank1987application}]\label{thm:lenstra}
	There is an algorithm that solves
	\begin{equation*}
		\max\{c\T x : Ax \le b, x\in \Z^n\}
	\end{equation*}
	in time $n^{O(n)} \poly(I)$ and polynomial space.
\end{theorem}
Note that~\eqref{eq:ilp} can be translated to the form in \Cref{thm:lenstra} without increasing the number
of variables, hence the same result holds there.
Together with the bound on $\Gamma$, this yields
a running time of $n^\Gamma \cdot \Gamma^\Gamma \cdot \poly(I) \le (m\Delta)^{O(m^2 \log{\Delta})} \cdot \poly(I)$ and polynomial space, using that without loss of generality all columns of $A$ are different and, in particular,
$n \le (2\Delta + 1)^m$.

For our result we introduce two new ideas.
In~\cref{sec:branching} we introduce the following branching algorithm that can be seen as an application of
Savitch's theorem to the support of the solution. Consider
the lexicographically minimal solution $x^\ast$ and let $x^\ast = x^\ast_\ell +
x^\ast_r$ be such that $x^\ast_\ell$ and $x^\ast_r$ have half the support each.
We will then guess a vector $b_\ell \in \Z^m$ such
that $Ax^\ast_\ell = b_\ell$ and for each such guess, recurs into subproblems
with the right-hand sides $b_\ell$ and $b - b_\ell$ and then retrieve both $x_\ell^\ast$
and $x_r^\ast$. Note, that in each recursive call, the support halves,
so the total depth of the recursion is $O(\log{\Gamma}) =
O(\log m + \log\log{\Delta})$. 
This idea already yields a $(m\Delta)^{O(m^2 \log\log{\Delta})} \poly(I)$ time and
polynomial space algorithm. However, the dependence on the number
of constraints $m$ is much higher than what is promised in~\cref{thm:main}.
The bottleneck of the earlier mentioned idea of guessing the support and applying Lenstra's algorithm
is that one enumerates $O(n^\Gamma)$ support sets.
In~\cref{sec:guessing} we show that the number of relevant support sets is much
lower and that they can be enumerated using the branching algorithm as a subroutine
on small instances where it runs more efficiently. This is then combined with running
Lenstra's algorithm on each support set.

\section{Branching algorithm}\label{sec:branching}
As a subroutine in~\cref{thm:main}, we will first introduce an algorithm based on the idea
of branching on the support, which has the following properties.
\begin{lemma}\label{lem:simple}
	Let $A \in \Z^{m \times n}$, $b \in \Z^m$, $c \in \Z^n$, and $\Delta \ge \norm{A}_\infty$.
	Assume that~\eqref{eq:ilp} is bounded and feasible and let $x^*$ be
	the lexicographically minimal optimal solution.

	For any $\sigma\in \mathbb N$ we can in time $(\sigma\Delta)^{O(m \log \Gamma)} \cdot \poly(n)$
	and polynomial space, compute $x^*$ under the promise that $\norm{x^*}_1 \le \sigma$. If this promise is not satisfied, there is no guarantee of correctness.
\end{lemma}
Note that the running time matches~\cref{thm:main} except for the dependence on $\sigma$.
There are known bounds on $h$, see e.g.,~\cite{papadimitriou1982complexity}, and
after the preprocessing in~\cite{jansen23} one can guarantee a bound of $\sigma = (m\Delta)^{O(m)}$.
This would therefore give a polynomial space algorithm with running time $(m\Delta)^{O(m^2\log\Gamma)}\poly(I)$,
but this is larger than the running time claimed in~\cref{thm:main}. We will
apply it more efficiently in \Cref{sec:guessing}.

\begin{algorithm}[ht!]
    \DontPrintSemicolon
	\textbf{function} $\text{branch}(A,b,c,\sigma,s)$:\\
	\If(\tcp*[f]{Base case}){$s \le 1$}{\Return 
lexicographically minimal optimal solution to  $\min \{c\T x : x \in \ZZ^n, Ax = b, \norm{x}_0 = s\}$ or $\bot$ if none exists} 
	$z \gets \bot$ \\
\ForEach{$b_\ell \in \{-\sigma\Delta,\dotsc,\sigma\Delta \}^m$}{\label{ln:guess}
		$x_\ell \gets \text{branch}(A,b_\ell,c,\sigma,\floor{s/2})$\\
		$x_r \gets \text{branch}(A,b - b_\ell,c,\sigma,\ceil{s/2})$\\
		\If{$x_\ell + x_r$ is \emph{better} than $z$}{
			$z \gets x_\ell + x_r$
		}
	}
	\Return $z$
	\caption{Pseudocode of the algorithm behind~\cref{lem:simple}. Input values $A,b,c,\sigma$ are defined as in~\cref{lem:simple}. Parameter $s\in\mathbb N$ denotes a bound on the support. The condition that $x_\ell + x_r$ is \emph{better} than $z$ means: 
		$x_\ell \neq \bot$ and $x_r \neq \bot$ and $(z = \bot$ or $c\T (x_\ell + x_r) < c\T z$ or $(c\T (x_\ell + x_r) = c\T z$ \text{ and } $x_\ell + x_r \lex z))$.
	}
	\label{alg:main}
\end{algorithm}
Now, we present the algorithm behind~\cref{lem:simple}. For the pseudocode, see~\cref{alg:main}.
This recursive algorithm has an additional parameter $s$, which is a bound on the support. 
On the initial call we use the bound $\Gamma$ from \cref{caratheodory}.
The base case consists of the
scenario when $s \le 1$. If $s = 0$, we return the all-zero vector if and only
if $b$ is an all-zero vector; otherwise, we return $\bot$. If $s = 1$ let $e_i$, $i
\in [n]$, be the indicator vector, i.e., a vector with 1
at position $i$ and 0 otherwise. We consider the set $\{z \cdot e_i : z \cdot A e_i = b, z
\in \Z, i \in [n]\}$ and return the last element
that minimizes the dot product with $c\T$ (or $\bot$ if no such element exists).
Observe that for each $i$, such an integer $z$ can be determined in $O(m)$
time; hence, in total, the base case can be decided in $O(mn)$ time.

When $s > 1$, we intuitively split the support of the solution into two (almost) equal-size parts
and guess the target $b_\ell$ for the first
part in the original solution. More precisely, we guess a
vector $b_\ell \in \{-\sigma\Delta, \ldots, \sigma\Delta \}^m$, recursively compute
$x_\ell$ and $x_r$, which are lexicographically minimal optimal solutions to $A x_\ell = b_\ell$ and $A x_r = b -
b_\ell$ with half the support of the previous call (rounded up for one recursion and down for the other),
and if both of them exist,
return $x_r + x_\ell$, breaking ties according to lexicographical order.
If all of the guesses for $b_\ell$ fail in at least one recursion, we return $\bot$.
This concludes the description of the algorithm.

\begin{proof}[Proof of~\Cref{lem:simple}]
For the running time, observe
that the integer $s$ decreases by a constant factor in each recursive call and starts with
	value $\Gamma$.
Thus, the depth of the recursion is $O(\log{\Gamma})$. The number of possible $b_\ell$
vectors at each recursion depth is $(2 \sigma\Delta + 1)^m$, so the total
size of the branching tree is $(\sigma\Delta)^{O(m \log{\Gamma})}$.
The base case takes $O(mn)$ time; hence, the total running time is as claimed.
	Similarly, the space complexity is $\poly(n, \log \Gamma)$ as the algorithm needs to store the
stack of the recursion of height $O(\log{\Gamma})$; the base case requires
	$\poly(n)$ space.

Next, we argue about the correctness. We prove by induction over $s$ that
$\text{branch}(A,b,c,\sigma,s)$ returns the lexicographically minimal optimal solution $x^* \in \ZZ^n$ to
\eqref{eq:ilp}, provided that $\norm{x^*}_1 \le \sigma$ and $\norm{x^*}_0 \le s$.

In the base case, $s \le 1$, the
algorithm straightforwardly checks over all possible non-zero coefficients of
$x^*$ and returns a lexicographically minimal one. Note that if $b$ is the all-zero vector then
assuming the integer program is bounded there cannot be a non-zero solution with negative objective
value. Hence, the all-zero solution is optimal and lexicographically minimal.

Now assume that $s > 1$. Let $x^\ast \in \ZZ^n$ be the lexicographically
minimal solution to~\eqref{eq:ilp} and assume that $\norm{x^*}_1 \le \sigma$ and $\norm{x^*}_0 \le s$. Let $x_\ell^\ast, x_r^\ast \in \ZZ^n$ be
arbitrary vectors such that $x_\ell^\ast + x_r^\ast = x^\ast$, $\norm{x_\ell}_0
\le \floor{\norm{x^\ast}_0/2}$, and $\norm{x_r^\ast}_0 \le
\ceil{\norm{x^\ast}_0/2}$.  For example, $x_\ell^\ast$ could be the first half
of the non-zero coefficients of $x^\ast$, and $x_r^\ast$ could be the remaining
non-zero coefficients.  Finally, let $b^\ast_\ell = Ax^\ast_\ell$. Observe
that $\norm{x^*_{\ell}}_1 \le \norm{x^*}_1 \le \sigma$, $\norm{x^*_{r}}_1 \le \norm{x^*}_1 \le \sigma$, and $b^\ast_\ell \in \{-\sigma\Delta,\dotsc,\sigma\Delta \}^m$.
Furthermore, $x^*_{\ell}$ and $x^*_{r}$ must be lexicographically minimal optimal solutions
for right-hand sides $b^\ast_{\ell}$ and $b - b^\ast_{\ell}$ since otherwise this would
contradict $x^\ast$ being the lexicographically minimal optimal solution for $b$.

Hence, if in~\cref{ln:guess} we guess $b_\ell = b_\ell^\ast$, the induction hypothesis guarantees
that $\text{branch}(A,b_\ell,c,\sigma,\floor{s/2})$ and
$\text{ILP}(A,b-b_\ell,c,\sigma,\ceil{s/2})$ return $x^*_{\ell}$ and $x^*_r$.
In particular, $x^* = x^*_{\ell} + x^*_r$ must be returned by the algorithm, since no
lexicographically smaller optimal solution exists that could dominate $x^*$.
\end{proof}

\section{Candidate support sets and main theorem}
\label{sec:guessing}

A naive enumeration of all support sets of cardinality at most $\Gamma$ would require
$O(n^\Gamma)$. In this section we will design a more efficient enumeration and prove \Cref{thm:main}.

\begin{definition}[Candidate support set]
	\label{def:candidates}
	For a matrix $A \in \Z^{m \times n}$, we say that a set $\mathcal{C} \subseteq \{0,1\}^n$ is a \emph{candidate support set} if:
	\begin{enumerate}[label=(\roman*)]
		\item\label{item:p1} $\norm{x}_1 \le \Gamma$ for every $x \in \mathcal{C}$, and
		\item\label{item:p2} for every $b \in \Z^m$ and $c \in \Z^n$ such that \eqref{eq:ilp} is bounded and feasible there exists an optimal solution $x^\ast$ to \eqref{eq:ilp} with $\supp(x^\ast) \in \mathcal{C}$.
	\end{enumerate}
\end{definition}

\begin{lemma}\label{lem:candidates}
	For every $A \in \ZZ^{m\times n}$ with $\norm{A}_\infty \le \Delta$, there exists a candidate support set of size $(m \Delta)^{O(m)}$. Moreover, this set can be enumerated in time $(m \Delta)^{O(m \log{\Gamma})} \cdot \poly(I)$ and polynomial space.
\end{lemma}

We note that, with \cref{lem:candidates}, the proof of \cref{thm:main} is a straightforward application of known results.

\begin{proof}[Proof of \cref{thm:main} assuming \cref{lem:candidates}]
	We enumerate a candidate support set $\mathcal{C}$ of $A$ using \cref{lem:candidates}. Then for every $c \in \mathcal{C}$, we consider the instance truncated to variables in $\supp(c)$ and solve it using the exact algorithm for integer programming by Lenstra and Kannan, i.e., \Cref{thm:lenstra}.

	The correctness follows from Property~\ref{item:p2} of $\mathcal{C}$. For the running time, note that by \cref{lem:candidates} we have $|\mathcal{C}| \le (m \Delta)^{O(m)}$ and it requires overhead of $(m \Delta)^{O(m \log \Gamma)} \poly(I)$ to enumerate $\mathcal C$. Recall that the algorithm of Lenstra and Kannan~\cite{kannan1983improved,lenstra1983integer} solves \eqref{eq:ilp} in $\Gamma^{O(\Gamma)} \poly(I)$ time and polynomial space for at most $\Gamma$ variables. In total, the running time of the algorithm is
	\begin{align*}
		(m \Delta)^{O(m \log \Gamma)} \poly(I) + |\mathcal{C}| \cdot \Gamma^{O(\Gamma)} \cdot \poly(I)
		&\le (m\Delta)^{O(m \log \Gamma)} \cdot \Gamma^{O(\Gamma)} \cdot \poly(I) \\
		&\le (m\Delta)^{O(m (\log m + \log\log \Delta))} \cdot \poly(I). \qedhere
	\end{align*}
\end{proof}
It remains to prove \cref{lem:candidates}.
This relies on the following property that Deng, Mao, and Zhong~\cite{DengMZ23} noticed for Unbounded Knapsack
and which easily generalizes to ILP.
\begin{lemma}\label{lem:2}
	Let $x^* \in \ZZ^n$ be the lexicographically minimum solution to \eqref{eq:ilp}. Then, there exists $x' \in \{0,1\}^n$ and $b' \in \Z^m$ such that (i) $x'$ is the lexicographically minimal optimal solution to $\min\{c\T x' : A x' = b'\}$ and (ii) $\supp(x') = \supp(x^*)$.
\end{lemma}
\begin{proof}
	Let $x'$ be the indicator vector of $\supp(x^*)$ and let $b' = Ax'$. Clearly, Property~(ii) holds. Hence, it remains to prove (i). For the sake of contradiction, assume that $y$ is a solution to the integer program with right-hand side $b'$ and $c\T y < c\T x'$
Consider $z := y + x^* - x'$. Then $z$ is a solution to $\{Ax = b', x\in\ZZ^n\}$. But $c\T z < c\T x^*$, which contradicts the optimality of $x^*$.
Similarly, if there is a solution $y$ with $c\T y = c\T x'$ and $y \lex x'$. Then
$z = y + x^* - x'$ is an optimal solution to $\min\{c\T x, Ax = b', x\in\ZZ^n\}$ that is lexicographically smaller than $x^*$.
\end{proof}
\begin{proof}[Proof of \cref{lem:candidates}]
First, we iterate over every $b' \in \ZZ^m$ such that $\norm{b'}_\infty \le \Gamma \cdot \Delta$. Note that the number of such $b'$ is $(2\Gamma \Delta + 1)^{m}$. Next, for each such $b'$, we invoke \cref{lem:simple} with $\sigma = \Gamma$, which returns the lexicographically minimal optimal solution $x^*$ to
\begin{equation*}
	\min\{c\T x : Ax = b', x\in\ZZ^n \}
\end{equation*}
	if $\norm{x^*}_1\le \sigma$.
	Finally, for each of these guesses, we output $x$ if $x \in \{0,1\}^n$. 
	This concludes the description of the algorithm. The algorithm runs in time
	\begin{equation*}
		(2\Gamma\Delta+1)^m \cdot (\Gamma \Delta)^{O(m \log \Gamma)} \cdot \poly(I) \le (m \Delta)^{O(m \log \Gamma)} \cdot \poly(I)
	\end{equation*}
	and uses polynomial space.
	Property~\ref{item:p1} of \cref{def:candidates} holds because every $x^*$ considered
	is lexicographically minimal and therefore by \cref{caratheodory} has support at most $\Gamma$. For Property~\ref{item:p2} consider some $b\in\Z^m, c\in \Z^n$ for which~\eqref{eq:ilp} is bounded and feasible. Let $x^*$ be the lexicographically minimal optimal solution. By \cref{caratheodory} 
	we have $\norm{x^*}_0 \le \Gamma$. Furthermore, by \Cref{lem:2} there is some $x'\in\{0,1\}^n$ with $\supp(x') = \supp(x^*)$, which is the lexicographically minimal optimal solution for $b' = Ax'$. Because $\norm{b'}_\infty \le \Gamma\Delta$ and $\norm{x'}_1 \le \sigma$ the algorithm above will consider $b'$ and correctly output $x'$.
\end{proof}

\section{Conclusion}
In this paper, we show that an integer linear program can be solved
in time $O(m\Delta)^{O(m(\log m + \log\log{\Delta})} \poly(I)$
and polynomial space.
It remains open whether the additional factor of $O(\log m + \log\log\Delta)$ in the exponent
compared to the algorithm without space limitation can be avoided.
We note that our algorithm has two bottlenecks that lead to this running time:
 (i) the depth of $O(\log m + \log\log{\Delta})$ for our branching and
(ii) because we use the algorithm of Lenstra and
Kannan that runs in time $N^{O(N)} \poly(I)$ on an instance with $N = O(m \log(m\Delta))$.
Thus, even if the improvement to $(\log{N})^{O(N)} \poly(I)$ by Reis and
Rothvoss~\cite{reis23} could be implemented in polynomial space, this would not immediately
lead to an improvement.
Already for the special case of Unbounded Knapsack it is unclear how to avoid
the exponent of $\log\log\Delta$ and achieve a pseudopolynomial running time.

Finally, we want to mention the case of binary integer linear programs, i.e., where
we substitude the domain constraint for $x\in\{0,1\}^n$ in~\eqref{eq:ilp}.
Here, the best known algorithm runs in time
$(m\Delta)^{O(m^2)}\poly(I)$~\cite{eisenbrand2019proximity} and naive application of
Savitch's theorem yields polynomial
space at significant higher running time of $(m\Delta)^{O(m^3 \log(m\Delta))}\poly(I)$.
It is unclear how to come close to
an exponent of $O(m^2)$ in polynomial space and the ideas in this paper do not seem to be
applicable, since the solutions need not be sparse.

\bibliographystyle{plain}
\bibliography{references}

\end{document}